%% file: arxiv-boost.tex
\documentclass[conference]{IEEEtran}
\ifCLASSINFOpdf
\else
\fi
\usepackage{tikz}
\usepackage{pgfplots}

\usepackage{ifpdf}
\usepackage{graphicx}
\usepackage{subfig}

\usepackage{url}
\usepackage{rotating}
\usepackage{multirow}

\usepackage{float}

\usepackage{algorithm}
\usepackage{algorithmic}

\usepackage{mdwmath}
\usepackage{mdwtab}
\usepackage{amsmath}

\newtheorem{definition}{Definition}

\newtheorem{lemma}{Lemma}

\newtheorem{theorem}{Theorem}

\def\MoBoox{\textsf{MoBoo}}
\def\MoBoo{\textsf{MoBoo }}
\def\MIIPx{\textsf{MIIP}}

\def\Greedy{\textsf{Greedy }}
\def\Greedyx{\textsf{Greedy}}
\def\GreedyBatch{\textsf{GreedyBatch }}
\def\GreedyBatchx{\textsf{GreedyBatch}}
\def\SPPx{\textsf{SPP-d}}

\def\SPTx{\textsf{SPT-d}}

\def\SPHx{\textsf{SPH-d}}

\def\LastNodex{\textsf{LastNode }}
\def\LastNode{\textsf{LastNode}}
\def\SPPMITx{\textsf{SPP-MIT}}
\def\SPPMIT{\textsf{SPP-MIT }}
\def\SPTMITx{\textsf{SPT-MIT}}

\def\SPHMITx{\textsf{SPH-MIT}}
\def\SPHMIT{\textsf{SPH-MIT }}

\def\TMoBoox{\textsf{TMoBoo}}
\def\TMoBoo{\textsf{TMoBoo }}
\def\FastTMoBoox{\textsf{FastTMoBoo}}

\def\MIIPs{\textsf{MIIPs }}
\def\TMIIPs{\textsf{TMIIPs }}

\def\SimpleMIIPsx{\textsf{MIIPs}}
\def\FastTMIIPsx{\textsf{FastTMIIPs}}
\def\SPTMTCITx{\textsf{SPT-MTCIT}}
\def\SPPMTCITx{\textsf{SPP-MTCIT}}
\def\SPHMTCITx{\textsf{SPH-MTCIT}}

\begin{document}

%
\title{
	Boosting Nodes for Improving the Spread of  Influence
}

\author{\IEEEauthorblockN{Konstantinos Liontis and Evaggelia Pitoura}
\IEEEauthorblockA{Computer Science and Engineering Department, University of Ioannina, Greece}
\{kliontis, pitoura\}@cs.uoi.gr}

\maketitle

\begin{abstract}
Information diffusion in networks has received a lot of recent attention. 
Most previous work addresses the influence maximization problem of selecting an appropriate set of seed nodes to initiate the diffusion process so that the largest number of nodes is reached. 
Since the seed selection problem is NP hard, most solutions are sub-optimal. 
Furthermore, there may be settings in which the seed nodes are predetermined. Thus, a natural question that arise is: given a set of seed nodes, can we select a small set of nodes such that if we improve their reaction to the diffusion process, the largest increase in diffusion spread is achieved?  We call this problem, the {\em boost set selection} 
problem. In this paper, we formalize this problem, study its complexity and propose appropriate algorithms. 
We also  evaluate the effect of boosting in a number of real networks and report the increase of influence spread achieved for different seed sets, time limits in the diffusion process and other diffusion parameters.
\end{abstract}


%
\IEEEpeerreviewmaketitle

\input{arxiv-intro}

\input{arxiv-related}
\input{arxiv-model}

\input{arxiv-algor}
\input{arxiv-experiments}
\input{arxiv-conclusion}
\bibliographystyle{abbrv}
\bibliography{boost}

\end{document}

%% file: arxiv-intro.tex
\section{Introduction}
Information diffusion
in social networks 
has received a lot of recent attention.
Among others, a central motivating application is viral marketing, where the adoption of 
a product or an idea by a small number of opinion leaders triggers a large cascade of further adoptions
via ``word of mouth''.
In this context, influence maximization is defined as the problem of identifying  a small 
set of initial nodes to influence such that the 
spread of influence in the network is maximized.

The theoretical and algorithmic foundations of  influence maximization
were laid in the seminal work of \cite{DR01} and \cite{KKT03}, where
the {\em Independent Cascade} (IC) model, one of the most commonly used propagation model, was also introduced.  
In this model, 
once influenced, a node has a single chance to influence its neighbors captured by
a per-edge influence probability.
Since then, there has been a surge of research activity focusing on many aspects of diffusion.
%
%
An important extension 
considers time in the diffusion process.
In real networks, diffusion is not spontaneous \cite{KKW08, LCXZ12}, instead,
there is delay or latency in the propagation of information from a person to another,
depending for example, on the activity patterns of users, e.g., how often
they check their social network accounts. Furthermore, for many applications,
time  is critical.
For example, it is important that a user hears about a concert before the concert takes place,
or, that opinions are formed before an election takes place.    

Most current research so far has focused on the influence maximization problem of selecting 
the influencers or early adopters, called {\em seeds}.
However, there are cases in which the initiators of an activation are  
predetermined, for example, seeds may correspond to news agencies, on-site reporters, eye-witnesses, 
or, known advocates of a product.
Furthermore, even when the seeds are not fixed, since the seed selection problem is NP hard \cite{CWY09}, often the set of the selected seeds 
is sub-optimal.
Thus, a natural question arises: How can we increase
the spread of diffusion for a given seed selection?


In this paper, 
we  assume that it is possible to improve the reaction to
the diffusion process of a small number of nodes by investing extra resources, e.g., by giving out free samples of a product, engaging gamification, or other marketing strategies.
In particular, we assume that we can make 
$k$ nodes more influential to others.
Then, for a given seed set,  
we would like to identify those {\em k} nodes whose improved reaction to a diffusion will result
in maximizing the average number of nodes influenced
by a given time instant $T$. 
We call this problem the {\em Boost Set Selection} problem.

Formally, we define boosting a node as improving its probability of influencing others,  
making the node react to an activation faster, or both. 
We show that the Boost Set Selection problem is NP-hard for the IC model.
Furthermore, we show that, in contrast to most influence maximization problems, the Boost Set Selection problem is not submodular.
Then, we propose a number of algorithms for the problem. 
The natural greedy algorithm runs at $k$ steps and at each step
selects to boost the node with the largest marginal benefit. However, greedy is computationally expensive, since multiple simulations are needed for
estimating the benefit  in the influence spread. We design efficient algorithms that
estimate the benefit of boosting through exploiting most probable paths.
Our first algorithm, the {\it Most Probable Path Algorithm} (\MoBoox) is a very fast
algorithm that ignores activation delays, while the {\it Time-dependent Most Probable Path Algorithm} (\TMoBoox) improves the approximation with some extra cost.
Our {\it Maximum Influence Independent Path} (\MIIPx) algorithms
consider additional diffusion propagation paths.
 Finally, we exploit a family of algorithms that select nodes
based on their proximity to the seed nodes.

Note that our complexity results and our algorithms are readily applicable to
an alternative interpretation of boosting in which improved reaction to an activation for a node means that the node itself is made more receptive to activations. 

We have evaluated the effect of boosting and the efficiency of our algorithms
experimentally using four real datasets. 
In many cases, boosting just a couple of nodes 
improves the diffusion spread more than adding more
seeds. 
Furthermore, our \MoBoo and \TMoBoo algorithms are orders of magnitude faster than greedy and return comparable boost sets.

The rest of this paper is structured as follows. In Section \ref{rw}, we place our work in context with related research. In Section \ref{mod}, we define the Boost Set Selection problem formally and present
results regarding its hardness, while in Section \ref{alg}, we introduce our algorithms. In 
Section \ref{exp}, we report our experimental results. Finally, Section \ref{con} offers conclusions.

%% file: arxiv-related.tex
\section{Related Work}
\label{rw}
Since the seminal work of \cite{DR01, KKT03}, 
there has been a large body of research on influence diffusion
including efficient algorithms, e.g., \cite{LKGFVG07,BBCL14,TXS14,TSXX15,DSGZ13,CWW10,KS06,LCXZ12,WCSX10,CWY09}, alternative formulations of the influence maximization problem and new propagation models, 
for instance, for capturing temporal aspects, e.g., \cite{LCXZ12,CLZ12,CT-IC}.
However, we are not aware of any other work directly related to the boosting problem as introduced in this paper.
%
Note that in this paper, we assume that the underlying diffusion graph is known.
There is a also large body of research on inferring the graph, e.g., \cite{GBS11,GBL10,GLS13}. 
Such research is clearly orthogonal to ours.

\noindent {\textbf {\textit{Temporal Aspects.}}} 
There have been various proposals for enhancing the IC model with temporal aspects. 
The time-aware propagation model employed in this paper is
most similar to the latency aware independent cascade model (LAIC) \cite{LCXZ12} that
incorporates latency information into the standard independent cascade model by using
a per node delay function.
Somewhat different approaches include 
a new independent cascade with meeting events (IC-M) model \cite{CLZ12}
and the continuously activated and time constrained IC model \cite{CT-IC}. Our model and algorithms are applicable to such models as well.

\noindent{\textbf {\textit {Algorithms for Seed Selection.}}} 
Since the greedy algorithm is computationally expensive, there have been numerous 
proposals of more efficient algorithms.
The most similar approach to our \MoBoo algorithm is PMIA \cite{CWW10}, that,  as \MoBoox,
assumes that influence propagates through the most probable paths so as to estimate the local influence of nodes for seed selection.
In this paper, we use this assumption for estimating the influence of boosting a node, which is a different problem that
leads to different algorithms. 

A line of other research focuses on improving the performance of greedy.
To this end, 
CELF \cite{LKGFVG07} exploits the submodularity property
to avoid re-evaluating the marginal gain of each candidate node at each iteration of greedy.
Since the boost selection problem is not submodular, CELF cannot be used in our case.
There are many other approaches to efficiently estimate influence.  
A run time optimal (up to a logarithmic factor) algorithm was recently proposed in \cite{BBCL14}, while CONTINEST 
uses randomizations for influence estimation in a continuous-time diffusion network \cite{DSGZ13}.
Two variations of the IC model were proposed in \cite{KS06} that require diffusion 
to be performed only through shortest paths in which case the greedy algorithms are faster.
In \cite{LCXZ12}, an improved greedy for the time-constrained influence maximization problem 
along with two influence spreading based algorithms are proposed. 
Other research  includes a community-based approach \cite{WCSX10} that exploits the structural properties of social networks.


Lastly, various fast algorithms for the IC model select seeds based on centrality measures such node degrees \cite{CWY09}, betweenness, or closeness.
Along this line, in this paper, we use the proximity of the boost nodes to the seeds which is a property relevant to our problem.

\vspace*{0.01in}
\noindent{\textbf {\textit{Edge Augmentation and Diffusion.}}}
Besides ``boosting'' existing connections, another way to improve diffusion is by  increasing the connectivity of the network through adding new edges. 
The authors of \cite{recommend} explore edges additions in the context of recommendations.
In particular, they ask which edges from a set of recommended edges if added to a network would result in maximizing  the content spread among all nodes.
This problem is different from the problem studied here since we look into improving the
probability of existing edges, instead of selecting new edges from a given set.
Furthermore, we select nodes instead of edges and assume that diffusion starts from specific initiators. 
The authors of \cite{gelling} ask which edges to add or remove from a network so as to speed-up or contain a dissemination.
They study epidemic disseminations where the focus is on affecting the epidemic threshold by  altering the leading eigenvalue of the adjacency matrix of the graph.  

%% file: arxiv-model.tex
\section{Problem Definition}
\label{mod}
In this section, we start by describing the time-constrained independent cascade propagation model
and then, we present a formal definition of the boost set selection problem, prove that
the problem is NP-hard and show that it is not submodular.  

\subsection{The Time-Constrained IC Model}
Let us first describe the standard Independent Cascade (IC)
influence propagation model as introduced in \cite{KKT03}. 
In the IC model, the underlying network is modeled as a directed graph $G(V, E)$, 
where $V$ is a set of nodes representing users and $E$ is a set of directed edges representing 
relationships between them. Each edge $(u, v)$ in $E$ is associated with an activation 
or influence probability $p_{uv} > 0$. 

Influence propagation proceeds in discrete steps. A seed set $S$ of nodes, $S$ $\subseteq$ $V$, 
is activated at step $0$, while all other nodes are inactive.
At any subsequent step  $i$ $\geq$ 1, any node that has become active at the previous $i-1$ step is given a single chance to activate any of its currently inactive neighbors.  A node $u$ succeeds in activating its neighbor $v$ with probability $p_{uv}$ independently of the history so far. Once activated, a node remains active. The diffusion process runs until no additional activations are possible.

For a set of nodes $A$, we call {\em influence spread} of $A$, denoted by $\sigma(A)$,  the expected number of activated nodes at the end of the process when $A$ is used as the seed.
Then, the {\it influence
maximization problem} is defined as the problem of finding, for a parameter $k$, a $k$-node set $S$ such that $\sigma(S)$ is maximized.

In the initial IC model, there was no notion of time. However, since 
the actual time of an activation is central in many applications, 
the initial model has been extended to incorporate time (e.g., \cite{CLZ12,LCXZ12,CT-IC}).
In most cases, users do not respond to an activation immediately.
Instead, when activated (e.g., notified about an item), they propagate the activation inside a period of time
whose duration depends on many factors such as their personal characteristics, or habits, such as, how often they check their
accounts, or their judgment about the immediacy of the item  \cite{KKW08,GBS11}.  


To model time-dependent diffusion,
in addition to the activation probability, we associate with each user $u$ a delay function
$d_u$ which captures the distribution of the activation delay of 
$u$. In particular, when  $u$ activates one of its neighbors,  
$d_u(t)$ is the probability that $u$ does so
at exactly $t$ time units after its own activation.
As in the IC model, initially, at time $0$, a seed set $S$ is activated. 
A node $u$ activated at time $t$ activates each of its inactive neighbors $v$ at $t+i$ with probability 
$p_{uv}$$d_{v}(i)$. A node $v$ may be activated by different neigbhors at different time instants; we assume the earliest   
amongst these time instants as the activation time of $v$.

We shall use $\sigma_T(S)$ to denote the expected number of activated nodes after $T$ time instances.  
The {\it time constrained influence maximization problem} is the problem of finding, for a given $k$ and $T$, 
a seed set $S$ with at most $k$ nodes such that $\sigma_T(S)$ is maximized, for a given $T$.

The influence maximization problem has been shown to be NP hard for both  the initial IC problem and  its various time-related extensions \cite{CLZ12,LCXZ12,CT-IC}. 

\subsection{The  Boost Set Selection Problem}
Often it may be possible to increase the influence spread by investing resources towards increasing
the ability of specific users to influence others. 
We call this process, {\em boosting}. 
We assume that  boosting a node results in increasing its probability of activating its neigbors 
as well as increasing the speed of this activation.
Specifically, we model the effect of boosting a node $u$ by altering the activation probabilities $p_{uv}$ 
and its delay function $d_u$.

The definition of the boost set selection problem is orthogonal to the
specifics of boosting, however, to make the description more 
concrete,  
we quantify the amount of ``boosting'' through a quantity $b$, 0 $<$ $b$ $\leq$ 1 that amounts for the increase in
the activation probabilities. For simplicity, we assume the same increase $b$ for all boosted nodes.

In terms of the delay function, we assume that boosting a node $u$ may also result in     		 
increased $d_u(t)$ values for small values of $t$, intuitively, shifting the delay distribution to the left.  That is, we assume that the probability of the boosted node to react early may also increase as an effect of boosting.

\begin{definition} {\sc (node boost)}
Boosting a node $u$ by $b$, 0 $<$ $b$ $\leq$ 1 results in:  
(1) replacing $p_{uv}$ for all edges $(u, v)$ $\in$ $E$, with $p'_{uv}$ where $p'_{uv}$ = $p_{uv}$ + $b$, if 
$p_{uv}$ + $b$ $\leq$ 1 and $p'_{uv}$ = 1, otherwise
and (2) replacing $d_u(t)$ with $d'_u(t)$, such that   
$\sum_{i=0}^{t}$$d'_u(i)$ $\geq$ $\sum_{i=0}^{t}$$d_u(i)$, for all $t \geq 0$. 
\end{definition}




We call {\em boosted influence spread},  $\pi_{S, T}(B)$, 
the expected number of nodes activated at time $T$
with a seed set $S$, if the set $B$ $\subseteq$ $V$ of nodes is boosted.
We are now ready to formulate the boost set selection problem.

\begin{definition}  {\sc (The Boost Set Selection Problem)}
Given a graph $G(V, E)$, 
a seed $S \subseteq V$ of initially activated nodes, 
find a set $B \subseteq V $ of $k$ nodes such that $\pi_{S, T}(B)$ 
is maximized. 
\end{definition}

\subsection{Hardness of Boost Set Selection}
In this section, we study the complexity of the boost set selection problem. We first show that the boost selection problem is NP hard. 

\begin{theorem}
The boost set  selection problem is NP-hard.
\end{theorem}

\begin{proof}
We prove the lemma by reduction from the Set Cover Problem. The Set Cover Problem is defined as follows.
Given a collection of subsets $X_1$, $X_2$ $\dots$ $X_m$ of a ground set $U$ = \{$x_1$, $x_2$, $\dots$ $x_n$\}, we  
ask whether there exist $k$ subsets whose union is equal to $U$. 
Given an arbitrary instance of the Set Cover problem, we define a corresponding bipartite graph $G$ as follows. There is a node $u_i$ for 
each set $X_i$, a node $u_j$ for each element $x_j$ and a directed edge $(u_i, u_j)$ with activation 
probability $p_{u_i,u_j}$ equal to 0. 
We assume a large enough $T$ such that the influence spread is not affected.
The sets $X_i$ correspond to the seeds.
The Set Cover problem is equivalent to deciding whether there is a boost set $B$ of $k$ nodes with $b$ = 1 in this graph with $\pi_{S, T}(B)$  $\geq$ $n+k$. 
If $X$ is a solution to the Set Cover problem,
then boosting the $k$ nodes that correspond to the subsets in $X$ will
result in changing their activation probabilities to 1, thus all nodes in the ground truth set will be activated.
Conversely, if for any set $B$ of $k$ nodes, it holds that $\pi_{S, T}(B)$  $\geq$ $n+k$, 
then this set is a solution to the Set Cover problem. \end{proof}

\begin{figure}
        \centering
                \includegraphics[width=5cm]{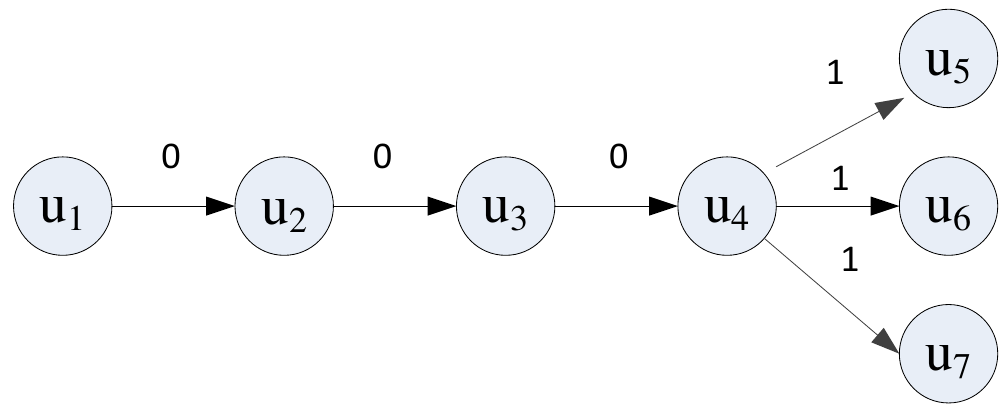}
        \caption{Counter example for showing non modularity.}
        \label{ce}
\end{figure}

Clearly, function $\pi_{S, T}(B)$  is monotonous, since
if we add a node to a boost set $B$, the mean number of nodes activated can only increase.
However, in contrast to most influence maximization problem, $\pi_{S, T}(B)$  is not submodular. 
Let $\Omega$ be a set of elements and $f$:$2^{\Omega}\rightarrow$ $R$ be a set function. 
A set function $f$ is called submodular, 
if $f(X \cup \{u\})$ - $f(X) \geq f(Y \cup \{u\}) - f(Y)$ for all elements $u$ in $\Omega$ and all pairs of sets $X \subseteq Y$.

\begin{theorem}
Function $\pi_{S, T}(B)$  is not submodular. 
\end{theorem}

\begin{proof}
We will prove the lemma through a counter example.
Consider the simple graph shown in Figure \ref{ce}.
The activation probabilities of all edges are equal to 0, except from the activation probabilities
of the outgoing edges of $u_4$ which are set equal to 1.
Assume that the seed set includes only node $u_1$, i.e., $S$ = $\{u_1\}$, time $T$ is large enough so that the spread
is not influenced and $b$ is 1.
Take boost sets $B_1$ = $\{u_1\}$ and $B_2$ = $\{u_1, u_2\}$. Clearly, $B_1 \subseteq B_2$.
Then, $\pi_{S, T}(B_1  \cup \{u_3\})$  - $\pi_{S, T}(B_1)$ = 1,
while $\pi_{S, T}(B_2  \cup \{u_3\})$  - $\pi_{S, T}(B_2)$ = 4, which shows that $\pi_{S, T}(B)$ is
not submodular. 
\end{proof}

%% file: arxiv-algor.tex
\section{Algorithms}
\label{alg}

In this section, we devise algorithms for selecting
the set of nodes to boost  so as to maximize the spread of influence.

\subsection{Basic Greedy}

A natural algorithm for the boost set selection problem is a greedy algorithm that works in $k$ steps as shown in
Algorithm \ref{alg:greedy}, selecting at each step to boost the node that 
causes the maximal marginal gain in influence spread (line 3).
\begin{algorithm}[t]
        \caption{General Greedy Algorithm}
        \small
        \label{alg:greedy}
        \begin{algorithmic}[1]
                \vspace{0.1cm}
                \STATE {initialize $B=\emptyset$}
                \FOR{$i$ = 1 \TO $k$} 
                        \STATE{$u$ =  $argmax_{\upsilon \in V \backslash B}$ ($\pi_{S, T}(B \cup \{\upsilon\})$ - $\pi_{S, T}(B))$  } 
                         \STATE $B$ = $B \cup \{u\}$ 
                \ENDFOR
        \end{algorithmic}
\end{algorithm}
The  gain is typically estimated by running a large number $R$ of simulations for each candidate node
as shown in Algorithm \ref{alg:greedy2}.
Graph $G^B$  is the graph   that results when the nodes in set $B$ are boosted,
that is, if we replace the activation probabilities 
and activation delays of the nodes in $B$ with the boosted probabilities and delays.
Instead of considering all nodes as candidates for boosting,
Algorithm \ref{alg:greedy2} works on the induced subgraph $G^c$ = $(V^c, E^c)$ of the original graph $G$ = $(V, E)$ that includes 
the nodes $V^c$ $\subseteq$ $V$ reachable within time $T$ from the nodes in the seed set $S$, if we assume  activation probabilities equal to one for all edges and 
zero delays for all nodes. Clearly, boosting nodes outside  $V^c$ does not affect
the spread, since these nodes are not reachable.\\
The computational complexity of Algorithm \ref{alg:greedy2} is O($knR$(cost of computing spread)), where $n$ is the number of nodes and $R$ is the number of simulations per node.
\begin{algorithm}[t]
	\caption{Simulation-based Greedy Algorithm}
	\small
	\label{alg:greedy2}
	\begin{algorithmic}[1]
		\vspace{0.1cm}
		\STATE {initialize $B=\emptyset$}
		\STATE{Let $V^c$ be the set of nodes reachable from $S$ within $T$}
		\FOR{$i$ = 1 \TO $k$}
		  \FORALL{$u$ in $V^c \backslash B$}
		    \STATE{s(u) = 0}
		    \FOR{$j$ = 1 \TO $R$}
		      \STATE{$s(u)$ += $compute\_\pi_{S, T}(G^{B \cup \{u\}})$}
		    \ENDFOR
		  \ENDFOR
		  \STATE{$u$ =  $argmax_{\upsilon \in V^c \backslash B} (s(\upsilon))$}
		  \STATE $B$ = $B \cup \{u\}$
		\ENDFOR
	\end{algorithmic}
\end{algorithm}
To compute the spread,  our Basic \Greedy or simply \Greedy uses the algorithm shown in Algorithm \ref{alg:greedy3} that is based on shortest path computations. 
To be able to use single source shortest path algorithms, we introduce a  virtual node $s^v$ and add edges from $s^v$ to all nodes $s$ in the seed set $S$
with activation probabilities equal to one and zero activation delays.
Algorithm \ref{alg:greedy3} estimates the number of activated nodes by computing shortest-path distances ($sp$) from this virtual seed $s^v$. The weight
of an edge $(u, v)$ corresponds to the delay introduced by the delay function $d_u$, while the edge $(u, v)$ is followed with probability $p_{uv}$.
A node is activated if its shortest path distance is smaller or equal to the time constraint $T$.
We use  $\Pi$ to denote the set of nodes that are active.\\
The cost of estimating the increase in spread using Dijkstra's algorithm is $O(mlogn)$ and thus the overall
complexity of \Greedy is $O(k n R  m logn)$.
In general \Greedy is prohibitively slow, especially since we cannot use
the CELF optimization \cite{LKGFVG07}. 
Thus, we also consider a simple variation of \Greedy
that runs only a single step of \Greedyx, orders nodes based on their gain and selects the top $k$ of them. We call this variation \GreedyBatchx. The complexity of \GreedyBatch is $O(n R  m logn)$.
Since the boost node selection problem is not submodular, we cannot deduce  the $(1-1/\epsilon)$ approximation bound for
the optimality of \Greedyx. However, \Greedy is still a natural algorithm that provides a comparison point for other approaches.

\begin{algorithm}[t]
        \caption{\Greedyx: $compute\_\pi_{S, T}(G^B$)}
        \small
        \label{alg:greedy3}
        \begin{algorithmic}[1]
        	\STATE{initialize $E^o$ = set of outgoing edges of $s^v$} 
        	\STATE{initialize active nodes $\Pi$ = $\emptyset$}
                \FORALL{$(u, v)$ in $E^o$}
					\STATE {initialize $sp[s^v, v]$ = 0}
					\STATE {initialize  $ap(s^v)$ = 1}
				\ENDFOR
                \WHILE{$E^o \not= \emptyset$ }
					\FORALL{$(u, v)$ in $E^o$}
						\STATE {draw $flag$ from Bernoulli($p_{u, v}$)}
						\IF{$flag = 0$}
						\STATE{ $E^o$ = $E^o \backslash \{(u, v)\}$}
						\ENDIF
						\STATE {draw flag $\delta(u)$ from $d_u$}
					\ENDFOR
					\STATE{$(u^a, v^a)$ =  $argmin_{(u, v) \in E^o}$ ($sp[s^v, u] + \delta(u))$} 
					\IF{$sp[s^v, u^a]$ + $\delta(u^a) <= T$ }
						\IF{$v^a$ not in $\Pi$}
			                \STATE{ $sp[s^v, v^a] = sp[s^v, u^a] + \delta(u^a)$}
			                \STATE{ $\Pi$ = $\Pi \cup \{v^a\}$}
							\FORALL{outgoing edges $(v^a, v)$ of $v^a$}
								\IF{$v$ not in $\Pi$}
					                \STATE{ $E^o$ = $E^o \cup \{v\}$}
								\ENDIF		
							\ENDFOR
						\ENDIF
					\ENDIF
	                \STATE{ $E^o$ = $E^o \backslash \{(u^a, v^a)\}$}
				\ENDWHILE
				\STATE{ return $|\Pi|$}
        \end{algorithmic}
\end{algorithm}

\begin{table}
	\begin{center}
		\caption{Input Parameters}
		\centering
		\small
		\begin{tabular}{| l | l | l |}
			\hline
			Description & Default & Range \\ \hline
			Time constraint ($T$) & 15  & 10-19   \\
			Size of boost set ($k$) & 5  & 1-10   \\
			Size of seed set ($|S|)$ & 2  &   1-10 \\
			Budget amount ($b$) & 0.1  &  0.01 - 0.225   \\
			\# of simulations $(R)$ & 10,000  &    \\
			Boost delay policy & 1st-tu  &  1st-tu,  \\
			&   &  2nd-tu, none \\
			Delay function $(d_{u})$ & Exponential  &  \\
			Exponential parameter $(\alpha)$  &  [0, 1] &   \\
			Model for assigning probabilities  & $wc$  &  Trivalency with \\ 
			&   &   0.05, 0.1, 0.15   \\ 
			&   &   based sets \\ 
			\# of independent paths ($\lambda$) & 2  &    \\
			\hline
		\end{tabular}
		\label{input}
	\end{center}
\end{table}

\subsection{Single Most Probable Activation Path Algorithm}
An effective way of approximating spread is
by assuming that influence propagation follows the most probable path \cite{CWW10}.
We apply this assumption to the boost set selection problem, first ignoring activation delays leading to a very fast algorithm and then improving the approximation by paying some extra computational cost
to incorporate time delays.

\noindent{\bf \em Ignoring Activation Delays.} 
Let us first define the propagation probability, $pp$, of a path $P$  = $(u_1, u_2,  \dots u_l)$,   
as $pp(P) = \prod_{i=1}^{l-1} p_{u_iu_{i+1}}$,
since  for an activation to be propagated through $P$, 
all nodes in $P$ need to be activated.
We denote with $t_{min}(P)$ the minimum
time for propagating an activation from $u_1$ to $u_l$.  
Then, for a graph $G(V, E)$, 
we define the maximum influence path, $MIP$, between two nodes $u$ and $v$ as: 
\begin{multline*}
MIP_{G, T}(u, v) = argmax_P \{pp(P) | P \text{ is a path from } u \text{ to } \\ v 
\text{ with } t_{min}(P) \leq T \text{ and  each subpath } 
P_w = (u, \dots, w)
\text{ of }  \\
P 
\text{ is a } MIP_{G, T}(u, w)\}
\end{multline*}


We assume that influence propagates through maximum influence paths starting from the seed nodes.
Thus, if $s^v$ is the virtual seed node, we assume that influence propagates through
a maximum influence tree ($MIT$) defined as:
\begin{equation*}
MIT_{G, T}(s^v) = \bigcup_{u \in V} MIP_{G, T}(s^v, u) 
\end{equation*}

We use the $MIT$ to estimate the gain of boosting a node.
The influence spread is equal to the sum of the activation probabilities $ap(w)$ of the 
nodes $w$ in the graph. The gain $g(u)$ in influence spread by boosting node $u$ is equal to the increase of these activation probabilities.
Note that if we boost a node, only the activation probabilities of its descendants in the MIT
are affected.
In our first algorithm, we assume that the activation
probability, $ap(w)$, of a node $w$ is equal to the propagation probability 
$pp$ of the single path in $MIT$ from $s^v$ to $w$.
Our algorithm uses the following lemma to estimate the gain.

\begin{lemma}
\label{lemma1}
Let $P$ = ($u_1$  $\dots$ $u_l$) be a path.
If we boost node $u_m$, $1$ $\leq$ $m<$ $l$, the gain in spread for the nodes of $P$ is equal to 
$(\frac{p'_{u_{m}u_{m+1}}} {p_{u_{m}u_{m+1}}} - 1)   \sum_{i=m+1}^{l}ap(u_i)$.
\end{lemma}

\begin{proof}
The propagation probability of
paths $P_1$ = ($u_1$ $\dots$ $u_m$) and $P_2$ = ($u_{m+1}$ $\dots$ $u_{l}$)
remains the same. 
Let  $w$ be a node in $P_2$, $ap'(w)$ be its activation probability after the boost and $P_w$ be the path from 
$u_{m+1}$ to $w$. It holds, $ap(w)$ = $ap(u_m)$ $p_{u_{m}u_{m+1}}$ $pp(P_w)$.
It also holds, $ap'(w)$ = $ap(u_m)$ $p'_{u_{m}u_{m+1}}$ $pp(P_w)$, since the probability of path $P_w$ does not change.
Thus, the gain for $w$ is: $ap'(w)$ - $ap(w)$ = ($\frac{p'_{u_{m}u_{m+1}}} {p_{u_{m}u_{m+1}}} - 1)$ $ap(w)$ which proves the lemma.
\end{proof}

From Lemma \ref{lemma1}, we can calculate the overall gain $g(u)$, that is the increase in spread, when we boost a node $u$. Let $OUT(u)$ denote the 
children of $u$ in the $MIT$.
\begin{equation}
g(u) = \sum_{v \, \in \, OUT(u)}(\frac{p'_{uv}} {p_{uv}} - 1)\sum_{w \, descedant \, of \,  v} ap(w)
\label{gain}
\end{equation}

We are now ready to describe the {\em Most Probable Path Boost} (\MoBoox) algorithm, shown in Algorithm \ref{alg:moboo}.
\MoBoo starts by constructing the maximum influence tree, $MIT$, using Dijkstra's shortest path algorithm with $s^v$ as the source node and as weight for an edge $(u, v)$ the probability $p_{uv}$.
During the construction of the $MIT$, we compute and store with each node $u$ 
its activation probability $ap(u)$.
This step of \MoBoo has complexity $O(mlogn)$.
Then, we compute for each node  $u$, the gain $g(u)$ attained if we boost $u$. The computation of gain is done in a single bottom-up traversal of 
$MIT$ that computes the gain for each node based on the activation
probabilities of its descendants using Equation (\ref{gain}). The complexity of this step is
$O(n)$.
Boosting a node may change the shortest paths for some nodes in
which case we may need to rebuilt the $MIT$.
Thus, \MoBoo has time complexity  $O(kmlogn+kn)$ and requires $O(n)$ storage.
Rebuilding the $MIT$ is expensive, especially for dense networks. Thus, we use the same $MIT$ and adjusting  the activation probabilities of the descendants of the boosted node. This variation
of \MoBoo has complexity $O(mlogn + kn)$.

\begin{algorithm}[t]
  \caption{Most Probable Path Boost (\MoBoox) Algorithm}
      \small
     \label{alg:moboo}
  \begin{algorithmic}[1]
          \vspace{0.1cm}
  \STATE {set $B = \emptyset$}
  \STATE {compute $MIT_{G,T}(s^v)$ and the activation probabilities}\label{moboo:line:seedMioaCon}
		\FOR{$i$ = 1 \TO $k$} 
              	\FORALL{$u$ in $MIT_{G,T}(s^v)$$\backslash B$}
	            	    \STATE {compute $g(u)$} 
      \ENDFOR
              \STATE{$u$ =  $argmax_{\upsilon \in MIT(s^v)\backslash B}$ $g(\upsilon)$ }
             
            \STATE{ $B$ = $B \cup \{u\}$}\label{moboo:line:joinB}
 \STATE{update $MIT_{G,T}(s^v)$ and the activation probabilities}
          \ENDFOR
 \end{algorithmic}
\end{algorithm}

\noindent{\bf \em Incorporating Activation Delays.}
\MoBoo is very efficient but ignores the distribution of the activation delays 
$d_u(t)$ associated with
each node $u$, thus underestimates the gain through 
less probable but fast paths.
If we take into account the activation delays, for a path $P$  = $(u_1, u_2,  \dots u_l)$, 
the  probability that node $u_l$ is activated by $u_1$ through $P$  is 
$ap_{T}(P) = pp(P)p_{T}(P)$, where with
$p_{T}(P)$,  we denote the probability that $u_l$ is activated within time $T$ from the time 
$u_1$ is activated. We discuss later how to compute $p_{T}(P)$. Let as assume for now that
the complexity of its computation is $C_T$.
For a graph $G(V, E)$, we define the maximum time constrained influence path, $MTCIP$, between two nodes $u$ and $w$
\begin{multline}
MTCIP_{G, T}(u, v) = argmax_P \{ap_{T}(P) | P \text{ is a path} \\
	\text{from } u 
\text{ to } v\} \nonumber
\end{multline}
and the maximum time constrained influence tree ($MTCIT$): $
MTCIT_{G, T}(s^v) = \bigcup_{u \in V} MTCIP_{G, T}(s^v, u) $.

We assume that the diffusion process unfolds through the $MTCIT$ 
and consider that the activation probability
$ap(u)$ of node $u$ 
is equal with the activation probability 
$ap_T(P)$ of the single path $P$ in $MTCIT$ from $s^v$ to $u$, that is,
with $pp(P)p_{T}(P)$.

As with \MoBoox, the Time-Constrained Most Probable Path Algorithm (\TMoBoox), shown in
Algorithm \ref{alg:Tmoboo}, first
creates the $MTCIT$, and then
iteratively discovers the node with the maximum gain.
However, we can no longer use Lemma \ref{lemma1} to compute the gain
for all candidate nodes in a single bottom-up traversal.
Instead, we need to compute the gain for each candidate node $u$ by aggregating for each descendant $w$ of $u$ the increase $g(u, w)$ in its activation probability if $u$   
is boosted.
To do this efficiently, Algorithm \ref{alg:Tmoboo} computes $g(u, w)$ for each node $w$, for each node
$u$ in the path $P_{w}$ from the
root $s^v$ to $w$, and
accumulates this gain for $u$.  
This computation of gain for all candidates takes $O(nhC_T)$ time at each step of the algorithm, where $h$ is the height of the $MTCIT$, resulting in an overall
$O(mlognC_{T} + knhC_{T})$ complexity.

\begin{algorithm}[t]
	\caption{Time-constrained Most Probable Path Boost (\TMoBoox) Algorithm}
	\small
	\label{alg:Tmoboo}
	\begin{algorithmic}[1]
		\vspace{0.1cm}
		\STATE {set $B = \emptyset$}
		\STATE {compute $MTCIT_{G, T}(s^v)$ and the activation probabilities}\label{alg:Tmoboo:createMTCIT}
		\FOR{$i$ = 1 \TO $k$}
		\FORALL{$w$ in $MTCIT_{G, T}(s^v)\ \backslash B$}
		\STATE{Let $P_{w}$ be the path from $s^v$ to $w$}\label{alg:Tmoboo:ap_Tw}
		\FORALL{$u$ in $P_{w}$}
	\STATE {compute $g(u, w)$} \label{alg:Tmoboo:ap_uTw}
	\STATE { $g(u) += g(u, w)$}\label{alg:Tmoboo:g_w_u}
		\ENDFOR
		\ENDFOR
		\STATE{$u$ =  $argmax_{\upsilon \in MTCIT_{G, T}(s^v)\backslash B}$ $g(\upsilon)$ }
		\STATE{ $B$ = $B \cup \{u\}$}\label{alg:mobooT:line:joinB}
		\STATE{update $MTCIT_{G, T}(s^v)$} and the activation probabilities \label{alg:Tmoboo:updateMTCIT}
		\ENDFOR
	\end{algorithmic}
\end{algorithm}
Let us estimate the $p_T(P)$ quantity for path $P$ = ($u_1$, $u_2$, $\dots$  $u_l$).
Let $U_1$, $U_2$, $\dots$ $U_{l-1}$ be random variables with density functions
equal to $d_{u_1}$, $d_{u_2}$, $\dots$ $d_{u_{l-1}}$ respectively
and let  $TP$ = $U_1 + U_2 + \dots + U_{l-1}$ be their sum, which is also a random variable.
Then,  $p_T(P)$ = $Pr(TP \leq T)$. Since $U_1$, $U_2$, $\dots$ $U_{l-1}$ are independent random variables,
we can compute $p_T(P)$ by taking the convolution of their density functions.
In general, computing $p_T(P)$ is expensive.
By discretizing the delay functions, given that $T$ is measured in time units, 
it requires $O(h^T)$, where $h$ is the height of the tree, i.e., the
longest path. In addition, we use a rough approximation of $p_T(P)$ 
with $Pr(U_{l-1} \leq \frac{T}{l-1})$ and call the respected algorithms fast.
\vspace{-2.5mm}
\subsection{Maximum Influence Independent Paths Boost Algorithm}

\MoBoo and \TMoBoo simplify the graph structure by constructing a tree with the most probable
paths 
ignoring the real network structure. For example, we connect each node with at most one seed,
or  assume that a node can be activated only by a single other node.
In this section, we describe how to achieve a better approximation, 
by using more than one influence propagation path.

We define two paths as \textit{independent} if they have the same destination node
and the only other node that they may have in common is their starting node.
Our goal is to create $\lambda$ independent paths for each node $w \in V$ such as the starting node will be one 
of the seeds and the ending node the node $w$. Additionally, we want these paths to capture the maximum 
influence of any of the seeds to $w$.
To this end, we need to be able to rank all the paths ending to a node starting from any of the seeds based on their ability to influence the ending node.
Until now, we have seen two methods in which we can choose one path
among a set of paths, used in $MIP$ and $MTCIP$ construction, which express the probability for the starting node to
activate the last node.
No matter the ranking method, denoted as $RM$, we can iteratively apply it $\lambda$ times
with the constraint introduced by the definition of the independent paths to choose the first $\lambda$ independent paths.
Let $P^{1}(s^{\upsilon},w)$ where
\begin{equation*}  
P^{1}(s^{\upsilon},w) = argmax_{P}\{RM(P) | P \text{ is a path from } s^{\upsilon} \text{ to } w\}
\end{equation*}  
be the first independent path from $s^{\upsilon}$ to $w$.
Then, the second path $P^{2}(s^{\upsilon},w)$ is created as: 
\begin{multline*}  
P^{2}(s^{\upsilon},w) = argmax_{P}\{RM(P) | P \text{ is a path from } s^{\upsilon} \text{ to } w 
\text{, } \\ 
P^{2}\cap P^{1} \subseteq (\{s^{\upsilon}\} 
\cup  S \cup \{ w\}) \}
\end{multline*}
where $P^{i} \cap P^{j}$ is the set of their common nodes.
In general, we define:
\begin{multline*}  
P^{i}(s^{\upsilon},w) = argmax_{P}\{RM(P) | P \text{ is a path from } s^{\upsilon} \text{ to } w 
\text{, } \\
P^{i}\cap (P^{1}\cup P^{2}\cup...P^{i-1}) \subseteq (\{s^{\upsilon}\} 
\cup  S \cup \{ w\}) \}
\end{multline*}
where $P^{i} \cup P^{j}$ is the union set of their nodes.
Finally, we define the $\lambda$ Most Ranking Independent Paths $MRIPs$ for node $w$ as:
\begin{equation*}  
\mu^{\lambda}(w) = \cup_{i=1}^{\lambda}P^{i}(s^{\upsilon},w)
\end{equation*}  

In this case, instead of a single tree, we maintain $\lambda$ independent paths per node
and use them to estimate the activation probabilities.
Assuming that every node $w$ can be activated only though its independent paths $\mu^{\lambda}(w)$,
its activation probability $ap(w)$ is estimated as:
\begin{equation*}  
ap(w) = 1 - \prod_{P^{i}(s^{\upsilon}, w)\ \in\ \mu^{\lambda}(w)}(1-RM(P^{i}(s^{\upsilon}, w)))
\end{equation*}


%
%

To estimate the gain, the {\em Maximum Independent Paths} algorithm first constructs 
 $\lambda$ independent paths per node and then uses them  
to evaluate the gain.
In particular,  it runs $k$ iterations. At each iteration, for each node $w$, it computes the gain $g(u, w)$ for each node
$u$ in  each $\mu^{\lambda}(w)$ and adds it to $g(u)$. At the end of the iteration,
it selects the node with the largest $g(u)$ and updates the paths.
The complexity depends on which $RM$ is used.
When $RM$ is based on the Most Probable Path (used in the $MIP$ construction), we need 
$O( n \lambda mlogn)$ time to build the paths, since we need to run the shortest path algorithm once per node. If the maximum length among all independent paths is $|P|$, the selection of
boost nodes requires $O(k n \lambda |P|)$. Thus, the overall time is $O(k n \lambda mlogn)$.
 We call this version of the algorithm \SimpleMIIPsx. 
When the $RM$ is based on the activation probability with time delays,
an additional cost $C_T$ is introduced resulting in an $O(k n \lambda mlogn C_{T})$
overall complexity. We call this version of the algorithm \TMIIPs.

\subsection{Algorithms based on Proximity to the Seeds}
A simple approach to the problem is to select as boost nodes those nodes that are the closest to the seed nodes.
The reason is that boosting these nodes affects a large number of other nodes.
We  propose three intuitive algorithms that differ on how distance is defined.
The {\em shortest-path time-based} algorithm (\SPTx) defines the distance between two nodes $u$ and $v$ based on the time delay of the  diffusion between them.
The {\em shortest-path probability-based} algorithm (\SPPx) defines the distance between two nodes $u$ and $v$ based on the propagation probability between them.
Finally, the {\em shortest-path hop-based} algorithm (\SPHx) defines the distance between two nodes $u$ and $v$ based on the number of hops on
the propagation path between them.
To estimate the distances, we run $R$ simulations.
At each simulation, the distance of a node from the seed nodes in $S$ is
calculated as
$distance(u, S)$ = $min_{s \in S} distance(u, s)$, where $distance$ is 
time-based, probability-based or hop-based depending on the algorithm employed.
Then, we select the top-$k$ nodes with the best $avg(distance(u, S))$.
The complexity of the path-based algorithms is 
$O(Rmlogn)$, since at each simulation, we use Dijkstra's algorithm to
compute the shortest path distances.

We also develop corresponding faster algorithms that use the maximum influence trees (the $MIT$ or, $MTCIT$) to estimate the
distance from the seeds 
again based on time, probability and hops  resulting in the \SPTMITx, \SPPMIT and \SPHMIT  (resp., \SPTMTCITx, \SPPMTCITx and \SPHMTCITx) algorithms.
For the \textsf{SPT} and \textsf{SPH} variants, we need to build the trees, thus their complexity is that of building the corresponding tree.
For the \textsf{SPP} variant, we do not need to build the 
whole tree; we stop as soon as, we find the first $k$ nodes, thus the complexity is 
$O(mlogk)$ (resp., $O(mlogkC_T)$).

Finally, we employ a {\em last-node} (\LastNode) algorithm that  selects to boost the
nodes at which diffusion stops most of the times. The intuition is that by boosting these nodes, the diffusion will continue and  additional nodes will be reached.
Again, we use $R$ simulations to locate these nodes. The complexity of this algorithm is $O(Rmlogn)$.
In case of ties, in all cases we select the nodes with the largest out-degree. 

%% file: arxiv-experiments.tex
\section{Experimental Evaluation}
\label{exp}
In this section, we present an evaluation of our approach. The goal of the evaluation is twofold.
First, we present a comparison of our algorithms for the boost selection problem in terms 
of their execution time and the quality of the computed solutions.
Second, we evaluate the effect of boosting for a variety of networks, seed selections and boosting parameters.
\subsection{Experimental Setup}
\noindent{\bf Input Parameters} For assigning probabilities, we use a weighted cascade model ($wc$), in which the probability $p_{uv}$ of an edge $(u, v)$ is set equal to $1/d$ where $d$ is the in-degree of node $v$ \cite{KKT03,CWW10,LCXZ12}.
We also use a trivalency model \cite{JHC12,CWW10,GBL2011}.
For each edge, we uniformly at random select a value from the set \{0.1, 0.01, 0.001\}, which corresponds to high, medium and low influences. 
For the delay function,we use an exponential model that is commonly used for modelling diffusion  \cite{ML2010,GRLK2012}. The  $\alpha$ parameter of the delay function for each node is selected 
uniformly at random from [0, 1].
We have experimented with various boosting parameters and report related experiments. As default, we use two seeds and select $k$ = 5 nodes for boosting. For the amount of boosting,we use $b$ = 0.1 as the default value. 
Boosting a node in addition to increasing the activation probabilities may  affect its delay function.
We use as default a ``1st-tu'' policy for speeding-up diffusion. With this strategy, we increase by $b$ the probability that a node responds within the 1st time unit.
We also use a ``2nd-tu'' strategy in which  we increase by $b$ the probability that a node responds within the 2nd time unit.
Other strategies are possible as well as long as the node reacts faster after boosting.
Our  input parameters are summarized in Table \ref{input}.

\begin{table*}
    \begin{minipage}{.7\linewidth}
    \caption{Execution time in ms (average)}
    \centering
	\begin{tabular}{ |c|c|c|c|c|c|c|c|c| }  	\hline
		&	 \multicolumn{2}{c|}{\textit{Wiki}} &	\multicolumn{2}{c|}{\textit{Epinions}} & \multicolumn{2}{c|}{\textit{Slashdot}} &
		\multicolumn{2}{c|}{\textit{DBLP}} 
		\\ \hline
		&	trv &	wc  &	trv	& wc	& trv &	wc &	trv &	wc \\
		\MoBoox &	254 &	80 &	918.5 &	1388 &	2700.2	& 4897 &	4.6	& 12 \\
		\TMoBoox &	10400 &	16026 &	1235364 &	806084 &	5690069	& 3370758 &	199	& 271 \\
		\FastTMoBoox &	4731 &	11833 &	504420 &	369973 &	2255074	& 1664951 &	128	& 150 \\
		\SimpleMIIPsx &	438365 &	103432 &	6173489 &	6240079 &	8768254	& 7048651 &	2610	& 2028 \\
		\FastTMIIPsx &	1021286 &	706734 &	25456842 &	24695487 &	51153953	& 51128653 &	6562	& 5310 \\
		\LastNode	& 5961.7 &	2585.5 &	8216.9	& 297.1	& 5313	& 436 &	788.6 &	105 \\
		\SPTMITx &	73.5 &	23.1 &	529.6	& 208.5	& 1125.9 &	497.5 &	1.6	 & 1 \\
		\SPPMITx &	0 &	0 &	0 &	0 &	0 &	0	& 0	 & 0 \\
		\SPHMITx &	78.4 &	24.2 &	568.2 &	200.8 &	1170 &	499.4 &	1.6	& 1.1 \\
		\SPTMTCITx &	4533 &	5195 &	8195	& 5236	& 50183 &	23365 &	74	& 7 \\
		\SPPMTCITx &	0 &	0 &	0	& 0	& 0 &	0 &	0	& 0 \\
		\SPHMTCITx &	4429 &	5063 &	8025	& 4865	& 48519 &	22574 &	132	&	5 \\ \hline
	\end{tabular}
	\label{alg-time}
    \end{minipage}%
    \begin{minipage}{.3\linewidth}
      	\centering
    	\caption{Execution time in hours (average) for the wc model}
		\begin{tabular}{ |c|c|c| }  	\hline
			& \Greedyx & \GreedyBatchx \\ \hline
			\textit{Wiki}&	34.5 &	5.26 \\
			\textit{Epinions} &	28.13 &	0.45 \\
			\textit{Slashdot} &	82.80	 & 13.05 \\
			\textit{DBLP} &	0.77 &	0.29 \\ \hline
		\end{tabular}
		\label{greedy-time}
      	\centering
    	\caption{Dataset Statistics}
    	\begin{tabular}{| l | l | l | l |}
    		\hline
    		& Nodes & Edges &  avg degree \\ \hline
    		\textit{Wiki} & 7,115 & 103,689 & 14.57\\ 
    		\textit{Epinions} & 75,879 & 508,837 & 6.705 \\ 
    		\textit{Slashdot} & 77,360 & 905,468 & 11.704 \\ 
    		\textit{DBLP} & 244,269 & 591,063 & 2.419 \\ 
    		\hline
   	\end{tabular}
	\label{datasets}
    \end{minipage} 
\end{table*}
%
\begin{figure*}
	\flushleft
	\subfloat[\textit{DBLP}]{
		\begin{tikzpicture}[scale=0.65]
		  \begin{axis}[xlabel=T,ylabel=SPREAD, ymin=6, ymax=17, xticklabels from table={dblpWC.txt}{T}, xtick=data]
		    \addplot[mark=square*, mark size = 2] table [x=T,y=NoBoost] {dblpWC.txt};
		    \addplot[blue,mark=square*,fill=white, mark size = 2] table [x=T,y=MoBoo] {dblpWC.txt};
		    \addplot[orange,mark=*,fill=white,mark size = 2] table [x=T,y=TMoBoo] {dblpWC.txt};
		    \addplot[red,mark=triangle*,fill=white,mark size = 2] table [x=T,y=FastTMoBoo] {dblpWC.txt};
		    \addplot[magenta,mark=star,fill=white,mark size = 2] table [x=T,y=MIIPs] {dblpWC.txt};
		    \addplot[cyan,mark=diamond,fill=white,mark size = 2] table [x=T,y=FastTMIIPs] {dblpWC.txt};
		    \addplot[olive,mark=oplus*,fill=white,mark size = 2] table [x=T,y=Greedy] {dblpWC.txt};
		    \addplot[violet,mark=Mercedes star,mark size = 2] table [x=T,y=GreedyBatch] {dblpWC.txt};
		  \end{axis}
		\end{tikzpicture}
	}
	\subfloat[\textit{Epinions}]{
		\begin{tikzpicture}[scale=0.65]
		  \begin{axis}[xlabel=T,ymin=3, ymax=25, xticklabels from table={epinionsWC.txt}{T}, xtick=data]
		    \addplot[mark=square*, mark size = 2] table [x=T,y=NoBoost] {epinionsWC.txt};
		    \addplot[blue,mark=square*,fill=white, mark size = 2] table [x=T,y=MoBoo] {epinionsWC.txt};
		    \addplot[orange,mark=*,fill=white,mark size = 2] table [x=T,y=TMoBoo] {epinionsWC.txt};
		    \addplot[red,mark=triangle*,fill=white,mark size = 2] table [x=T,y=FastTMoBoo] {epinionsWC.txt};
		    \addplot[magenta,mark=star,fill=white,mark size = 2] table [x=T,y=MIIPs] {epinionsWC.txt};
		    \addplot[cyan,mark=diamond,fill=white,mark size = 2] table [x=T,y=FastTMIIPs] {epinionsWC.txt};
		    \addplot[olive,mark=oplus*,fill=white,mark size = 2] table [x=T,y=Greedy] {epinionsWC.txt};
		    \addplot[violet,mark=Mercedes star,mark size = 2] table [x=T,y=GreedyBatch] {epinionsWC.txt};
		  \end{axis}
		\end{tikzpicture}
	}
	\subfloat[\textit{Slashdot}]{
		\begin{tikzpicture}[scale=0.65]
		  \begin{axis}[xlabel=T,ymin=4, ymax=55, xticklabels from table={slashdotWC.txt}{T}, xtick=data,
			       legend cell align=left, legend columns=1, legend style={at={(1.05,0.64)},anchor=west,font=\normalsize, mark size=13pt}]
		    \addplot[mark=square*, mark size = 2] table [x=T,y=NoBoost] {slashdotWC.txt};
		    \addplot[blue,mark=square*,fill=white, mark size = 2] table [x=T,y=MoBoo] {slashdotWC.txt};
		    \addplot[orange,mark=*,fill=white,mark size = 2] table [x=T,y=TMoBoo] {slashdotWC.txt};
		    \addplot[red,mark=triangle*,fill=white,mark size = 2] table [x=T,y=FastTMoBoo] {slashdotWC.txt};
		    \addplot[magenta,mark=star,fill=white,mark size = 2] table [x=T,y=MIIPs] {slashdotWC.txt};
		    \addplot[cyan,mark=diamond,fill=white,mark size = 2] table [x=T,y=FastTMIIPs] {slashdotWC.txt};
		    \addplot[olive,mark=oplus*,fill=white,mark size = 2] table [x=T,y=Greedy] {slashdotWC.txt};
		    \addplot[violet,mark=Mercedes star,mark size = 2] table [x=T,y=GreedyBatch] {slashdotWC.txt};
       		    \addlegendentry{NoBoost}
		    \addlegendentry{MoBoo}
		    \addlegendentry{TMoBoo}
		    \addlegendentry{FastTMoBoo}
		    \addlegendentry{MIIPs}
		    \addlegendentry{FastTMIIPs}
		    \addlegendentry{Greedy}
		    \addlegendentry{GreedyBatch}
		  \end{axis}
		\end{tikzpicture}
	}
	\caption{Spread of the spread-based approximation algorithms for various datasets and the wc model}
	\label{wc-spreadBased-Datasets}
\end{figure*}
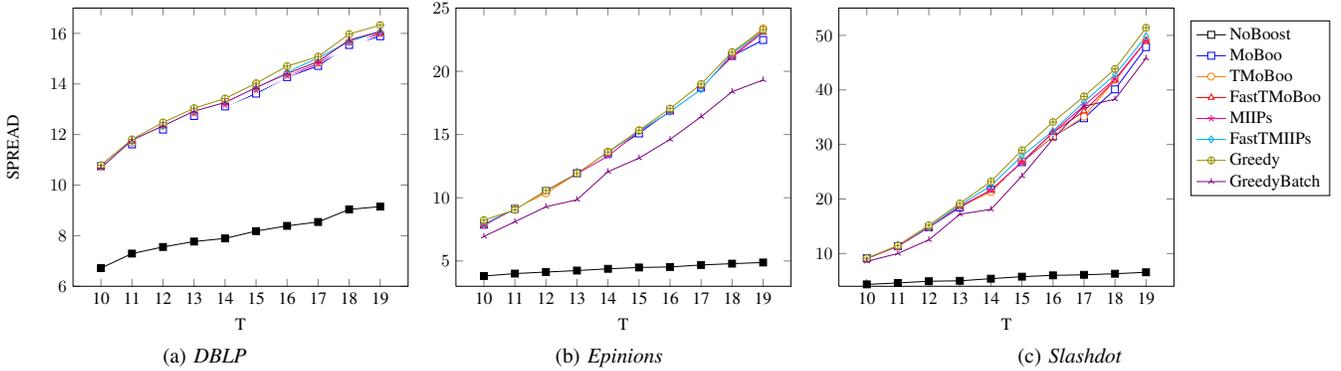

\begin{figure*}
	\flushleft
	\subfloat[\textit{Spread approximation algorithms}]{
		\begin{tikzpicture}[scale=0.65]
		  \begin{axis}[xlabel=T,ylabel=SPREAD, ymin=8, ymax=70, xticklabels from table={wikiWC.txt}{T}, xtick=data,
			       legend cell align=left, legend columns=2, legend style={at={(0.25,0.3)},anchor=west,font=\normalsize, mark size=13pt, draw=none}]
		    \addplot[mark=square*, mark size = 2] table [x=T,y=NoBoost] {wikiWC.txt};
		    \addplot[blue,mark=square*,fill=white, mark size = 2] table [x=T,y=MoBoo] {wikiWC.txt};
		    \addplot[orange,mark=*,fill=white,mark size = 2] table [x=T,y=TMoBoo] {wikiWC.txt};
		    \addplot[red,mark=triangle*,fill=white,mark size = 2] table [x=T,y=FastTMoBoo] {wikiWC.txt};
		    \addplot[magenta,mark=star,fill=white,mark size = 2] table [x=T,y=MIIPs] {wikiWC.txt};
		    \addplot[cyan,mark=diamond,fill=white,mark size = 2] table [x=T,y=FastTMIIPs] {wikiWC.txt};
		    \addplot[olive,mark=oplus*,fill=white,mark size = 2] table [x=T,y=Greedy] {wikiWC.txt};
		    \addplot[violet,mark=Mercedes star,mark size = 2] table [x=T,y=GreedyBatch] {wikiWC.txt};
      		    \addlegendentry{NoBoost}
		    \addlegendentry{MoBoo}
		    \addlegendentry{TMoBoo}
		    \addlegendentry{FastTMoBoo}
		    \addlegendentry{MIIPs}
		    \addlegendentry{FastTMIIPs}
		    \addlegendentry{Greedy}
		    \addlegendentry{GreedyBatch}
		  \end{axis}
		\end{tikzpicture}
	}
	\subfloat[\textit{Distance based algorithms}]{
		\begin{tikzpicture}[scale=0.65]
		  \begin{axis}[xlabel=T,ymin=5, ymax=70, xticklabels from table={wikiWC.txt}{T}, xtick=data,
			       legend cell align=left, legend columns=2, legend style={at={(0.21,0.35)},anchor=west,font=\normalsize, mark size=13pt, draw=none}]
		    \addplot[mark=square*, mark size = 2] table [x=T,y=NoBoost] {wikiWC.txt};
		    \addplot[orange,mark=*,fill=white,mark size = 2] table [x=T,y=TMoBoo] {wikiWC.txt};
		    \addplot[blue,mark=square*,fill=white, mark size = 2] table [x=T,y=SPP-MTCIT] {wikiWC.txt};
		    \addplot[red,mark=triangle*,fill=white,mark size = 2] table [x=T,y=SPT-MTCIT] {wikiWC.txt};
		    \addplot[magenta,mark=star,fill=white,mark size = 2] table [x=T,y=SPH-MTCIT] {wikiWC.txt};
		    \addplot[cyan,mark=diamond,fill=white,mark size = 2] table [x=T,y=SPP-MIT] {wikiWC.txt};
		    \addplot[olive,mark=oplus*,fill=white,mark size = 2] table [x=T,y=SPT-MIT] {wikiWC.txt};
		    \addplot[violet,mark=Mercedes star,mark size = 2] table [x=T,y=SPH-MIT] {wikiWC.txt};
		    \addplot[green,mark=diamond,mark size = 2] table [x=T,y=LastNode] {wikiWC.txt};
      		    \addlegendentry{NoBoost}
		    \addlegendentry{TMoBoo}
		    \addlegendentry{SPP-MTCIT}
		    \addlegendentry{SPT-MTCIT}
		    \addlegendentry{SPH-MTCIT}
		    \addlegendentry{SPP-MIT}
		    \addlegendentry{SPT-MIT}
		    \addlegendentry{SPH-MIT}
		    \addlegendentry{LastNode}
		  \end{axis}
		\end{tikzpicture}
	}
	\subfloat[\textit{Spread approximation algorithms - TRV}]{
		\begin{tikzpicture}[scale=0.65]
		  \begin{axis}[xlabel=T,ymin=10, ymax=90, xticklabels from table={wikiTRV.txt}{T}, xtick=data,
			       legend cell align=left, legend columns=1, legend style={at={(0.55,0.43)},anchor=west,font=\normalsize, mark size=13pt, draw=none}]
		    \addplot[mark=square*, mark size = 2] table [x=T,y=NoBoost] {wikiTRV.txt};
		    \addplot[blue,mark=square*,fill=white, mark size = 2] table [x=T,y=MoBoo] {wikiTRV.txt};
		    \addplot[orange,mark=*,fill=white,mark size = 2] table [x=T,y=TMoBoo] {wikiTRV.txt};
		    \addplot[red,mark=triangle*,fill=white,mark size = 2] table [x=T,y=FastTMoBoo] {wikiTRV.txt};
		    \addplot[magenta,mark=star,fill=white,mark size = 2] table [x=T,y=MIIPs] {wikiTRV.txt};
		    \addplot[cyan,mark=diamond,fill=white,mark size = 2] table [x=T,y=FastTMIIPs] {wikiTRV.txt};
       		    \addlegendentry{NoBoost}
		    \addlegendentry{MoBoo}
		    \addlegendentry{TMoBoo}
		    \addlegendentry{FastTMoBoo}
		    \addlegendentry{MIIPs}
		    \addlegendentry{FastTMIIPs}
		  \end{axis}
		\end{tikzpicture}
	}
	\caption{Spread of the \textit{Wiki} dataset}
	\label{spread-wiki}
\end{figure*}
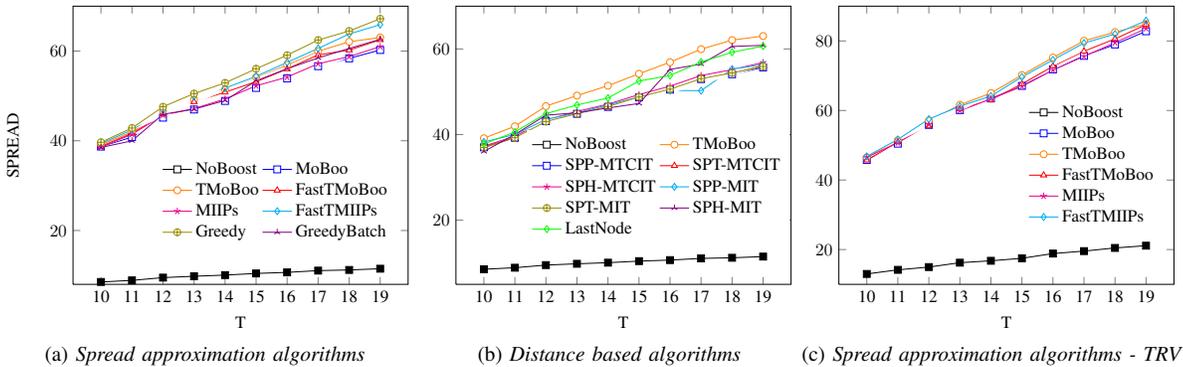


\begin{figure*}
	\flushleft
	\subfloat[\textit{DBLP}]{
		\begin{tikzpicture}[scale=0.65]
		  \begin{axis}[xlabel=T,ylabel=SPREAD, ymin=6, ymax=17, xticklabels from table={dblpWC.txt}{T}, xtick=data]
		    \addplot[mark=square*, mark size = 2] table [x=T,y=NoBoost] {dblpWC.txt};
		    \addplot[orange,mark=*,fill=white,mark size = 2] table [x=T,y=TMoBoo] {dblpWC.txt};
		    \addplot[blue,mark=square*,fill=white, mark size = 2] table [x=T,y=SPP-MTCIT] {dblpWC.txt};
		    \addplot[red,mark=triangle*,fill=white,mark size = 2] table [x=T,y=SPT-MTCIT] {dblpWC.txt};
		    \addplot[magenta,mark=star,fill=white,mark size = 2] table [x=T,y=SPH-MTCIT] {dblpWC.txt};
		    \addplot[cyan,mark=diamond,fill=white,mark size = 2] table [x=T,y=SPP-MIT] {dblpWC.txt};
		    \addplot[olive,mark=oplus*,fill=white,mark size = 2] table [x=T,y=SPT-MIT] {dblpWC.txt};
		    \addplot[violet,mark=Mercedes star,mark size = 2] table [x=T,y=SPH-MIT] {dblpWC.txt};
		    \addplot[green,mark=diamond,mark size = 2] table [x=T,y=LastNode] {dblpWC.txt};
		  \end{axis}
		\end{tikzpicture}
	}
	\subfloat[\textit{Epinions}]{
		\begin{tikzpicture}[scale=0.65]
		  \begin{axis}[xlabel=T,ymin=2, ymax=24, xticklabels from table={epinionsWC.txt}{T}, xtick=data]
		    \addplot[mark=square*, mark size = 2] table [x=T,y=NoBoost] {epinionsWC.txt};
		    \addplot[orange,mark=*,fill=white,mark size = 2] table [x=T,y=TMoBoo] {epinionsWC.txt};
		    \addplot[blue,mark=square*,fill=white, mark size = 2] table [x=T,y=SPP-MTCIT] {epinionsWC.txt};
		    \addplot[red,mark=triangle*,fill=white,mark size = 2] table [x=T,y=SPT-MTCIT] {epinionsWC.txt};
		    \addplot[magenta,mark=star,fill=white,mark size = 2] table [x=T,y=SPH-MTCIT] {epinionsWC.txt};
		    \addplot[cyan,mark=diamond,fill=white,mark size = 2] table [x=T,y=SPP-MIT] {epinionsWC.txt};
		    \addplot[olive,mark=oplus*,fill=white,mark size = 2] table [x=T,y=SPT-MIT] {epinionsWC.txt};
		    \addplot[violet,mark=Mercedes star,mark size = 2] table [x=T,y=SPH-MIT] {epinionsWC.txt};
		    \addplot[green,mark=diamond,mark size = 2] table [x=T,y=LastNode] {epinionsWC.txt};
		  \end{axis}
		\end{tikzpicture}
	}
	\subfloat[\textit{Slashdot}]{
		\begin{tikzpicture}[scale=0.65]
		  \begin{axis}[xlabel=T,ymin=3, ymax=50, xticklabels from table={slashdotWC.txt}{T}, xtick=data,
			       legend cell align=left, legend columns=1, legend style={at={(1.05,0.64)},anchor=west,font=\normalsize, mark size=13pt}]
		    \addplot[mark=square*, mark size = 2] table [x=T,y=NoBoost] {slashdotWC.txt};
		    \addplot[orange,mark=*,fill=white,mark size = 2] table [x=T,y=TMoBoo] {slashdotWC.txt};
		    \addplot[blue,mark=square*,fill=white, mark size = 2] table [x=T,y=SPP-MTCIT] {slashdotWC.txt};
		    \addplot[red,mark=triangle*,fill=white,mark size = 2] table [x=T,y=SPT-MTCIT] {slashdotWC.txt};
		    \addplot[magenta,mark=star,fill=white,mark size = 2] table [x=T,y=SPH-MTCIT] {slashdotWC.txt};
		    \addplot[cyan,mark=diamond,fill=white,mark size = 2] table [x=T,y=SPP-MIT] {slashdotWC.txt};
		    \addplot[olive,mark=oplus*,fill=white,mark size = 2] table [x=T,y=SPT-MIT] {slashdotWC.txt};
		    \addplot[violet,mark=Mercedes star,mark size = 2] table [x=T,y=SPH-MIT] {slashdotWC.txt};
		    \addplot[green,mark=diamond,mark size = 2] table [x=T,y=LastNode] {slashdotWC.txt};
      		    \addlegendentry{NoBoost}
		    \addlegendentry{TMoBoo}
		    \addlegendentry{SPP-MTCIT}
		    \addlegendentry{SPT-MTCIT}
		    \addlegendentry{SPH-MTCIT}
		    \addlegendentry{SPP-MIT}
		    \addlegendentry{SPT-MIT}
		    \addlegendentry{SPH-MIT}
		    \addlegendentry{LastNode}
		  \end{axis}
		\end{tikzpicture}
	}
	\caption{Spread of the distance-based algorithms for various datasets and the wc model}
	\label{wc-seedDistance-Datasets}
\end{figure*}
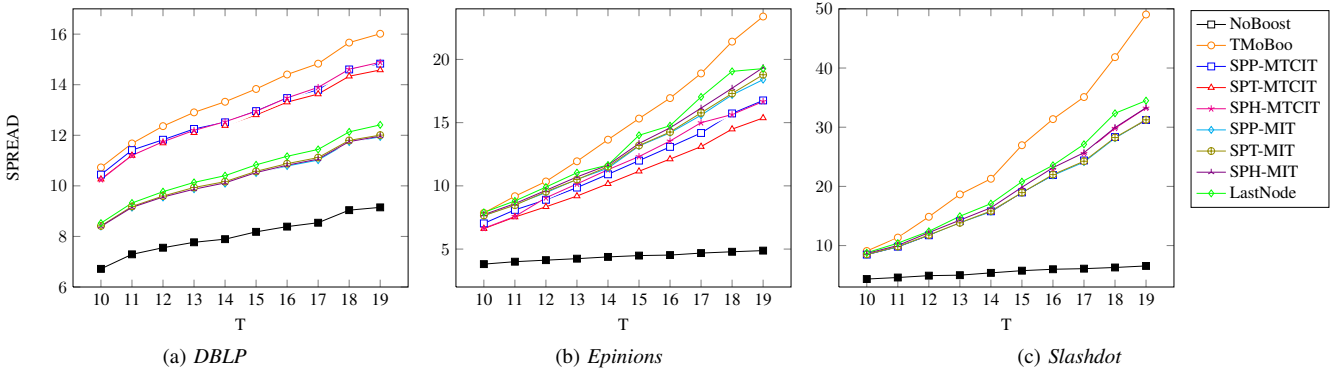

We use the following real datasets\footnote{\url{http://snap.stanford.edu},\url{http://dblp.uni-trier.de/db/}}: (1) {\it Wiki}
where an edge between two Wikipedia users indicates that a user was voted by another user, 
(2) {\it Epinions},  a who-trust-whom social network of the consumer review site Epinions.com, (3) {\it Slashdot},  a social network from the Slashdot technology-related news website, and (4) a subset of
{\it DBLP} where there is a bi-directional edge between two nodes if they were co-authors of an article in 2009 or 2010.  
The basic characteristics of the four datasets are summarized in Table \ref{datasets}.


\subsection{Comparison of the Algorithms}

In the first set of experiments, we compare the various algorithms for the boost selection problem.
We group the proposed algorithms in two categories: (a) those based on approximating the calculation of spread 
through most probable paths (\MoBoox, \TMoBoox, \FastTMoBoox, \SimpleMIIPsx, \FastTMIIPsx, \Greedyx, \GreedyBatchx) and 
those based on the distance from the seed (\SPTMITx, \SPPMITx, \SPHMITx, \SPTMITx, \SPTMTCITx, \SPPMTCITx, \SPHMTCITx, \LastNodex).
We do not report results for \SPTx, \SPTx, \SPHx\ as their spread is much lower than all the others. 

In Figures \ref{wc-spreadBased-Datasets} and \ref{spread-wiki}(a) we report the achieved spread with the wc model of the algorithms based on the spread approximation, while in 
Figures \ref{spread-wiki}(b) and \ref{wc-seedDistance-Datasets} the achieved spread of the algorithms based on the distance from the seed. 
In addition, we present for comparison, the  spread when no nodes are boosted (denoted as $NoBoost$).
We also experimented with a random selection of boost nodes, as well as, with selecting the nodes with the largest degree: both algorithms have spread comparable to that of $NoBoost$.

As expected, as we use more sophisticated spread estimation
algorithms, we get improved spread, however, the improvement
is small. The improvement increases when the networks are dense.
By taking into account time delays, \TMoBoox \, and \FastTMoBoox \,
achieve better spread than the attractive \MoBoox\ by paying a significant increase in time complexity.
The approximation used by \FastTMoBoox \, to estimate time delays is very effective.
\MIIPs that uses multiple paths but ignores delays is in generally less effective and much slower than \FastTMoBoox \,, indicating that incorporating time is more important than
considering additional influence propagation paths.
Regarding the distance-based algorithms, the \TMoBoox\ algorithm outperforms them in all cases.
Among them, \LastNode \, (an implementation of the physical meaning of ``increase the capacity of a bottleneck conduit'') appears to achieve the larger spread.
In general, algorithms based on distance fail to handle cases where boosting nodes further from a seed may result in better spread, for example,
by connecting some other seed with a larger number of nodes, or due to the distribution of the probabilities.
Furthermore, their performance is not consistent for different datasets, since they are based on rough estimations.
Figure \ref{spread-wiki}(c) compares the spread approximation algorithms for the trivalency model.
The relative performance of the algorithms remains the same but the differences are even smaller.

Table \ref{alg-time} depicts the running time of all algorithms. 
\MoBoo is much faster than the simulation based algorithms.
Under the trivalency model, the difference is more than one order of magnitude. 
With the $wc$ model, the execution time of \MoBoo increases for the large
datasets. The reason is that the $wc$ model assigns on average smaller activation probabilities on edges.
Thus, simulation-based algorithms are faster, since the spread is smaller.
Comparing \MoBoo and \MIIPs, we notice the expected large difference in their time complexity.
Incorporating the delay functions increases the running time at least one order of magnitude.
\TMIIPs algorithm can not be used in practice and its fast version needs at least five times more time than its simple version.
The \SPPMITx\ (\SPPMTCITx) has negligible execution time, since it does not
wait for the full $MIT$ ($MTCIT$) construction, but terminates when $k$ nodes added to the $MIT$.
Because the execution time of \Greedy is significant large (Table \ref{greedy-time}), we conducted experiments using only the weighted cascade model. 
Notice also  that the running time of \Greedy is not just $k$ times larger than the running time of \GreedyBatch.
This is because, any new node that joins the boost set results in increasing the spread for next candidate and the time needed for selecting a new node is always larger
than that of the previous selection.

\subsection{Evaluation of Boosting}
In this set of experiments, we present results regarding the various parameters of boosting.
For computing the set of nodes to boost, we use the \TMoBoox \ algorithm. We run our experiments on all datasets 
but report results for the \textit{Wiki} and \textit{Epinions} datasets, since the results for the other two datasets are similar.\\
\begin{figure*}
	\flushleft
	\subfloat[\textit{Wiki-Number of seeds}]{
		\begin{tikzpicture}[scale=0.65]
		  \begin{axis}[xlabel=\# OF SEEDS,ylabel=SPREAD-(\textit{Wiki}), ymin=5, ymax=135, xticklabels from table={wiki_Seeds_k_Spread.txt}{OfSeeds}, xtick=data]
		    \addplot[cyan,mark=diamond,fill=white,mark size = 2] table [x=OfSeeds,y=k10] {wiki_Seeds_k_Spread.txt};
		    \addplot[magenta,mark=star,fill=white,mark size = 2] table [x=OfSeeds,y=k7] {wiki_Seeds_k_Spread.txt};
		    \addplot[red,mark=triangle*,fill=white,mark size = 2] table [x=OfSeeds,y=k4] {wiki_Seeds_k_Spread.txt};
		    \addplot[blue,mark=square*,fill=white, mark size = 2] table [x=OfSeeds,y=k2] {wiki_Seeds_k_Spread.txt};
		    \addplot[orange,mark=*,fill=white,mark size = 2] table [x=OfSeeds,y=k1] {wiki_Seeds_k_Spread.txt};
		    \addplot[mark=square*, mark size = 2] table [x=OfSeeds,y=k0] {wiki_Seeds_k_Spread.txt};
		  \end{axis}
		\end{tikzpicture}
	}
	\subfloat[\textit{Wiki-Time constraint $T$}]{
		\begin{tikzpicture}[scale=0.65]
		  \begin{axis}[xlabel=T, ymin=7, ymax=73, xticklabels from table={wiki_T_k_Spread.txt}{T}, xtick=data]
		    \addplot[cyan,mark=diamond,fill=white,mark size = 2] table [x=T,y=k10] {wiki_T_k_Spread.txt};
		    \addplot[magenta,mark=star,fill=white,mark size = 2] table [x=T,y=k7] {wiki_T_k_Spread.txt};
		    \addplot[red,mark=triangle*,fill=white,mark size = 2] table [x=T,y=k4] {wiki_T_k_Spread.txt};
		    \addplot[blue,mark=square*,fill=white, mark size = 2] table [x=T,y=k2] {wiki_T_k_Spread.txt};
		    \addplot[orange,mark=*,fill=white,mark size = 2] table [x=T,y=k1] {wiki_T_k_Spread.txt};
		    \addplot[mark=square*, mark size = 2] table [x=T,y=k0] {wiki_T_k_Spread.txt};
		  \end{axis}
		\end{tikzpicture}
	}
	\subfloat[\textit{Wiki-Budget amount $b$}]{
		\begin{tikzpicture}[scale=0.65]
		  \begin{axis}[xlabel=BUDGET AMOUNT, ymin=7, ymax=185,
			       legend cell align=left, legend columns=1, legend style={at={(1.05,0.64)},anchor=west,font=\normalsize, mark size=13pt}]
		    \addplot[opacity=0.5, thick, cyan,mark=diamond,fill=white,mark size = 2] table [x=b,y=k10] {wiki_b_k_Spread.txt};
		    \addplot[opacity=0.5, thick, magenta,mark=star,fill=white,mark size = 2] table [x=b,y=k7] {wiki_b_k_Spread.txt};
		    \addplot[opacity=0.5, thick, red,mark=triangle*,fill=white,mark size = 2] table [x=b,y=k4] {wiki_b_k_Spread.txt};
		    \addplot[opacity=0.5, thick, blue,mark=square*,fill=white, mark size = 2] table [x=b,y=k2] {wiki_b_k_Spread.txt};
		    \addplot[opacity=0.5, thick, orange,mark=*,fill=white,mark size = 2] table [x=b,y=k1] {wiki_b_k_Spread.txt};
		    \addplot[opacity=0.5, thick, mark=square*, mark size = 2] table [x=b,y=k0] {wiki_b_k_Spread.txt};
		  \end{axis}
		\end{tikzpicture}
	}
	\\
	\flushleft
	\subfloat[\textit{Epinions-Number of seeds}]{
		\begin{tikzpicture}[scale=0.65]
		  \begin{axis}[xlabel=\# OF SEEDS,ylabel=SPREAD-(\textit{Epinions}), ymin=4, ymax=75, xticklabels from table={epinions_Seeds_k_Spread.txt}{OfSeeds}, xtick=data]
		    \addplot[cyan,mark=diamond,fill=white,mark size = 2] table [x=OfSeeds,y=k10] {epinions_Seeds_k_Spread.txt};
		    \addplot[magenta, mark=star,fill=white,mark size = 2] table [x=OfSeeds,y=k7] {epinions_Seeds_k_Spread.txt};		    
		    \addplot[red,mark=triangle*,fill=white,mark size = 2] table [x=OfSeeds,y=k4] {epinions_Seeds_k_Spread.txt};
		    \addplot[blue,mark=square*,fill=white, mark size = 2] table [x=OfSeeds,y=k2] {epinions_Seeds_k_Spread.txt};
		    \addplot[orange, mark=*,fill=white,mark size = 2] table [x=OfSeeds,y=k1] {epinions_Seeds_k_Spread.txt};
		    \addplot[mark=square*, mark size = 2] table [x=OfSeeds,y=k0] {epinions_Seeds_k_Spread.txt};
		  \end{axis}
		\end{tikzpicture}
	}
	\subfloat[\textit{Epinions-Time constraint $T$}]{
		\begin{tikzpicture}[scale=0.65]
		  \begin{axis}[xlabel=T, ymin=3, ymax=30, xticklabels from table={epinions_T_k_Spread.txt}{T}, xtick=data]
		    \addplot[cyan,mark=diamond,fill=white,mark size = 2] table [x=T,y=k10] {epinions_T_k_Spread.txt};
		    \addplot[opacity=0.5, thick, magenta,mark=star,fill=white,mark size = 2] table [x=T,y=k7] {epinions_T_k_Spread.txt};
		    \addplot[red,mark=triangle*,fill=white,mark size = 2] table [x=T,y=k4] {epinions_T_k_Spread.txt};
		    \addplot[blue,mark=square*,fill=white, mark size = 2] table [x=T,y=k2] {epinions_T_k_Spread.txt};
		    \addplot[orange,mark=*,fill=white,mark size = 2] table [x=T,y=k1] {epinions_T_k_Spread.txt};
		    \addplot[mark=square*, mark size = 2] table [x=T,y=k0] {epinions_T_k_Spread.txt};
		  \end{axis}
		\end{tikzpicture}
	}
	\subfloat[\textit{Epinions-Budget amount $b$}]{
		\begin{tikzpicture}[scale=0.65]
		  \begin{axis}[xlabel=BUDGET AMOUNT, ymin=3, ymax=75,
			       legend cell align=left, legend columns=1, legend style={at={(1.05,0.64)},anchor=west,font=\normalsize, mark size=13pt}]
		    \addplot[opacity=0.5, thick, cyan,mark=diamond,fill=white,mark size = 2] table [x=b,y=k10] {epinions_b_k_Spread.txt};
		    \addplot[opacity=0.5, thick, magenta,mark=star,fill=white,mark size = 2] table [x=b,y=k7] {epinions_b_k_Spread.txt};
		    \addplot[opacity=0.5, thick, red,mark=triangle*,fill=white,mark size = 2] table [x=b,y=k4] {epinions_b_k_Spread.txt};
		    \addplot[opacity=0.5, thick, blue,mark=square*,fill=white, mark size = 2] table [x=b,y=k2] {epinions_b_k_Spread.txt};
		    \addplot[opacity=0.5, thick, orange,mark=*,fill=white,mark size = 2] table [x=b,y=k1] {epinions_b_k_Spread.txt};
		    \addplot[opacity=0.5, thick, mark=square*, mark size = 2] table [x=b,y=k0] {epinions_b_k_Spread.txt};
       		    \addlegendentry{k10}
       		    \addlegendentry{k7}
       		    \addlegendentry{k4}
       		    \addlegendentry{k2}
       		    \addlegendentry{k1}
       		    \addlegendentry{k0}
		  \end{axis}
		\end{tikzpicture}
	}
	\caption{Results of varying the influence parameters for the \textit{Wiki} and \textit{Epinions} Dataset}
	\label{fig:exp:Wiki-Epinions_1}
\end{figure*}
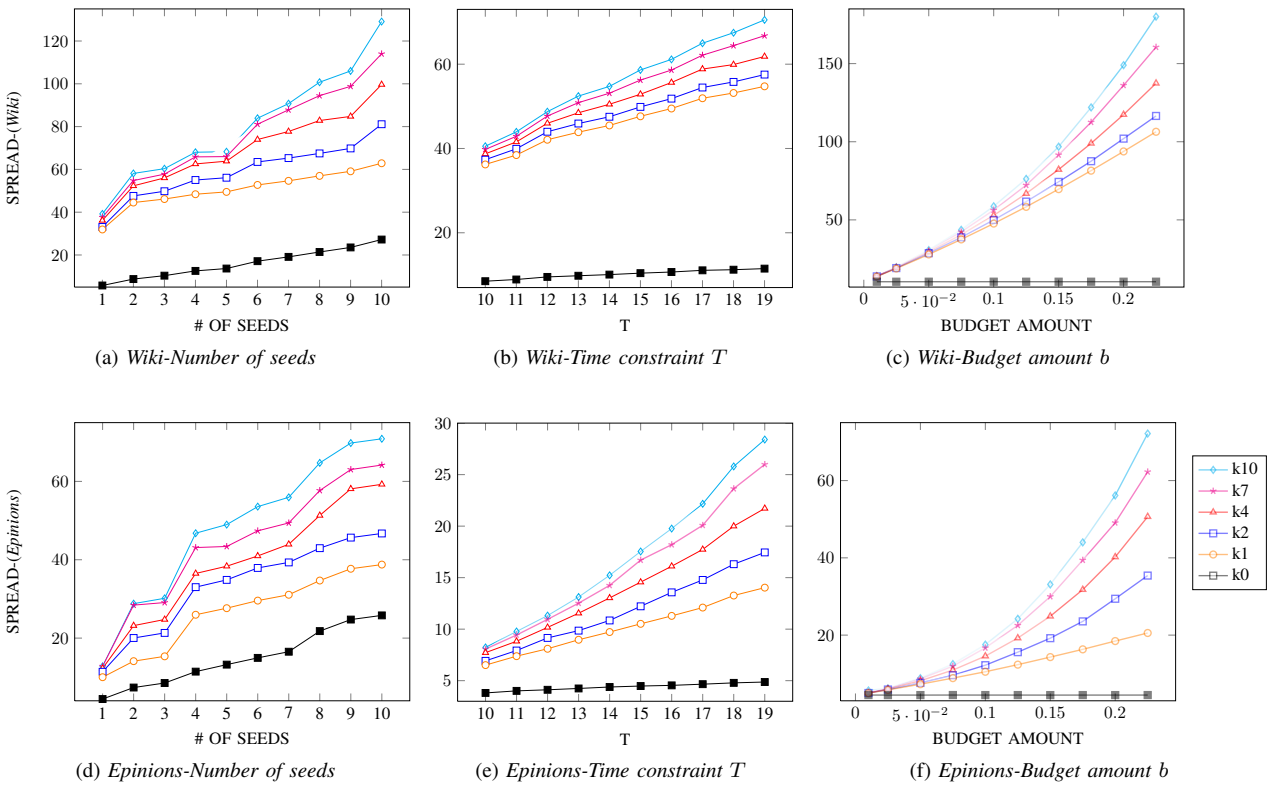
\vspace{-2em}
\begin{figure*}
	\flushleft
	\subfloat[\textit{Type of seed selection}]{
		\begin{tikzpicture}[scale=0.65]
		  \begin{axis}[xlabel=T,ylabel=SPREAD, ymin=2, ymax=84, xticklabels from table={wiki_T_seedsQuality_Spread.txt}{T}, xtick=data,
			       legend cell align=left, legend columns=2, legend style={at={(0.21,0.61)},anchor=west,font=\normalsize, mark size=11pt, draw=none}]
		    \addplot[mark=square*, mark size = 2] table [x=T,y=PoorNoBoost] {wiki_T_seedsQuality_Spread.txt};
		    \addplot[orange,mark=square*,fill=white,mark size = 2] table [x=T,y=Poor] {wiki_T_seedsQuality_Spread.txt};
		    \addplot[blue,mark=triangle*,mark size = 2] table [x=T,y=MediumNoBoost] {wiki_T_seedsQuality_Spread.txt};
		    \addplot[red,mark=triangle*,fill=white,mark size = 2] table [x=T,y=Medium] {wiki_T_seedsQuality_Spread.txt};
		    \addplot[magenta,mark=*,mark size = 2] table [x=T,y=GoodNoBoost] {wiki_T_seedsQuality_Spread.txt};
		    \addplot[cyan,mark=*,fill=white,mark size = 2] table [x=T,y=Good] {wiki_T_seedsQuality_Spread.txt};
       		    \addlegendentry{PoorNoBoost}
       		    \addlegendentry{Poor}
       		    \addlegendentry{MediumNoBoost}
       		    \addlegendentry{Medium}
       		    \addlegendentry{GoodNoBoost}
       		    \addlegendentry{Good}
		  \end{axis}
		\end{tikzpicture}
	}
	\subfloat[\textit{Trivalency models}]{
		\begin{tikzpicture}[scale=0.65]
		  \begin{axis}[xlabel=T, ymin=7, ymax=145, xticklabels from table={wiki_T_TRVVariance_Spread.txt}{T}, xtick=data,
			       legend cell align=left, legend columns=2, legend style={at={(0.01,0.86)},anchor=west,font=\normalsize, mark size=11pt, draw=none}]
		    \addplot[mark=square*,mark size = 2] table [x=T,y=TRV005NoBoost] {wiki_T_TRVVariance_Spread.txt};
		    \addplot[orange,mark=square*, fill=white,mark size = 2] table [x=T,y=TRV005] {wiki_T_TRVVariance_Spread.txt};
		    \addplot[blue,mark=triangle*,mark size = 2] table [x=T,y=TRV010NoBoost] {wiki_T_TRVVariance_Spread.txt};
		    \addplot[red,mark=triangle*,fill=white, mark size = 2] table [x=T,y=TRV010] {wiki_T_TRVVariance_Spread.txt};
		    \addplot[magenta,mark=*, mark size = 2] table [x=T,y=TRV015NoBoost] {wiki_T_TRVVariance_Spread.txt};		    
		    \addplot[cyan,mark=*,fill=white,mark size = 2] table [x=T,y=TRV015] {wiki_T_TRVVariance_Spread.txt};
      		    \addlegendentry{TRV005NoBoost}
      		    \addlegendentry{TRV005}
       		    \addlegendentry{TRV010NoBoost}
       		    \addlegendentry{TRV010}
       		    \addlegendentry{TRV015NoBoost}       		    
       		    \addlegendentry{TRV015}
		  \end{axis}
		\end{tikzpicture}
	}
	\subfloat[\textit{Speeding-up policy variance}]{
		\begin{tikzpicture}[scale=0.65]
		  \begin{axis}[xlabel=T, ymin=4, ymax=82, xticklabels from table={wiki_T_bAbStr_Spread.txt}{T}, xtick=data,
			       legend cell align=left, legend columns=2, legend style={at={(0.0,0.87)},anchor=west,font=\normalsize, mark size=11pt, draw=none}]
		    \addplot[orange,mark=diamond*, fill=white,mark size = 2] table [x=T,y=b0.1PropProbOnly] {wiki_T_bAbStr_Spread.txt};
		    \addplot[magenta,mark=diamond*, mark size = 2] table [x=T,y=b0.2PropProbOnly] {wiki_T_bAbStr_Spread.txt};		    
		    \addplot[blue,mark=triangle*,fill=white, mark size = 2] table [x=T,y=b0.1-1st-tu] {wiki_T_bAbStr_Spread.txt};
		    \addplot[cyan,mark=triangle*, mark size = 2] table [x=T,y=b0.2-1st-tu] {wiki_T_bAbStr_Spread.txt};
 		    \addplot[red,mark=*,fill=white, mark size = 2] table [x=T,y=b0.1-2nd-tu] {wiki_T_bAbStr_Spread.txt};
 		    \addplot[red,mark=*,mark size = 2] table [x=T,y=b0.2-2nd-tu] {wiki_T_bAbStr_Spread.txt};
		    \addplot[mark=square*,mark size = 2] table [x=T,y=NoBoost] {wiki_T_bAbStr_Spread.txt};		    
      		    \addlegendentry{b0.1PropProbOnly}
      		    \addlegendentry{b0.2PropProbOnly}
       		    \addlegendentry{b0.1-1st-tu}
       		    \addlegendentry{b0.2-1st-tu}
       		    \addlegendentry{b0.1-2nd-tu}
       		    \addlegendentry{b0.2-2nd-tu}
      		    \addlegendentry{NoBoost}
		  \end{axis}
		\end{tikzpicture}
	}
	\caption{Varying various boosting parameters for the \textit{Wiki} Dataset}
	\label{fig:exp:Wiki-Vote_2}
\end{figure*}
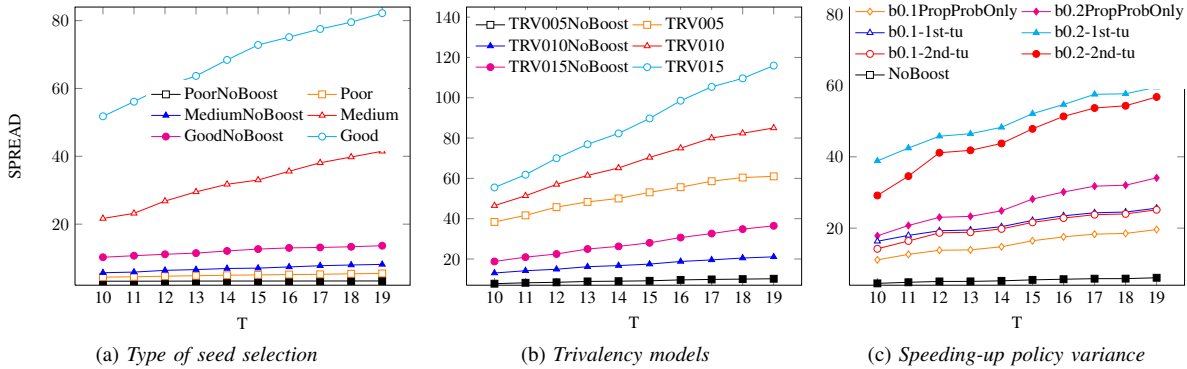

\noindent{\bf Varying the number of seeds.} Figures \ref{fig:exp:Wiki-Epinions_1}(a) and \ref{fig:exp:Wiki-Epinions_1}(d)
depict the spread attained with the number of seeds for various values of the number $k$ of boosted nodes.
We can see that in many cases boosting just a couple of nodes results in a larger increase in spread than selecting more seeds. Another observation is that in general,
the best gain is when the number of boosted nodes is close to the number of seed nodes.
Boosting additional nodes still increases the spread but the marginal gain decreases.
One reason is that most often the best node to boost is the seed. In this case, the gain is large, since the seed influences a large number of other nodes.\\
%
\noindent{\bf Varying the time constraint.} Figures \ref{fig:exp:Wiki-Epinions_1}(b) and \ref{fig:exp:Wiki-Epinions_1}(e) depict the spread with the time constraint $T$ for various values of the number $k$ of boosted nodes for two seed nodes.
This experiment shows that as the spread increases with larger values of $T$, the difference by
boosting additional nodes increases even when we just have one seed node.
\\
\noindent{\bf Varying the amount of budget.} Figures \ref{fig:exp:Wiki-Epinions_1}(c) and \ref{fig:exp:Wiki-Epinions_1}(f) depict the spread with the budget amount $b$ for various values of the number $k$ of boosted nodes. 
Clearly, the larger the $b$, the more the increase of spread. This indicates that strategies for motivating users to react to activations more often is as important as motivating influential initiators.
\vspace*{0.05in}
We also performed experiments with different seed selections, trivalency models and budget allocation strategies. For these  experiments, we report results for the Wiki dataset. The results for the other datasets are qualitative the same and are skipped.\\
\noindent{\bf Seed selection.} In this experiment, we select different types
of seeds (good, medium and poor) and report in Figure \ref{fig:exp:Wiki-Vote_2}(a) the effect of boosting in each of these cases.
In all cases, boosting increases the spread, showing that boost may be useful independently of the quality of the initial seed selection. 
\noindent{\bf Trivalency model.} We use two additional  trivalency models using for trivalency values the set \{0.05, 0.005, 0.0005\}  (denoted TR005) and the set \{0.15, 0.015, 0.0015\} (denoted TR015). We use TR010 to denote the default model. The results are depicted in Figure \ref{fig:exp:Wiki-Vote_2}(b).
Boosting works well for all cases.
An interesting observation is that by boosting just a few nodes of a less active network (i.e., one with small activation probabilities), we achieve the same spread as in the case of a very active network without boosting.\\
\noindent{\bf Policies for altering the delay function $d_u$.} We experimented with different ways of speeding up the reaction of a node, that is, with different ways of modifying the delay function.
The results are shown in Figure \ref{fig:exp:Wiki-Vote_2}(c).
With PropProbOnly, we denote the strategy where we change only the propagation probability of the node and do not alter the delay function.
The strategy ``1st-tu'' increases the probability of the 1st time unit,
while the strategy ``2nd-tu'' of the 2nd time unit.
We consider two values of $b$, namely, $b$ = 0.1 and $b$ =  0.2.
Clearly, the ``1st-tu'' strategy achieves the best spread and this is even more evident for the larger value of $b$.

%% file: arxiv-conclusion.tex
\section{Conclusions}
\label{con}
Information diffusion and propagation have attracted a lot of attention. A large body of related research focuses on influence spread maximization
by identifying a small set of influential nodes or seeds to initialize the diffusion. In this paper, we introduce a new problem.
We look into identifying the set of nodes whose improved (e.g., more frequent or more rapid) reaction to the diffusion process would result in maximizing the spread for a given set of initiator nodes.
We formalize the problem, study its complexity and present algorithms for its solution.
Our experimental results show that boosting a small set of nodes results in improving the diffusion spread, often more significantly than adding a few seeds.